\documentclass[a4paper,10pt]{article}

\usepackage[utf8]{inputenc}
\usepackage{amssymb,amsmath,amsthm,mathrsfs,enumitem,hyperref,pifont,pgfplots}
\usepackage{authblk}
\usepackage[style=trad-abbrv,doi=false,isbn=false,url=false,backend=biber,hyperref=true]{biblatex}
\usepgfplotslibrary{fillbetween}
\hypersetup{
    colorlinks,
    citecolor=blue,
    linkcolor=blue,
}
\usepackage{titling}
 \predate{}
 \postdate{}
\usepackage[margin=3cm]{geometry}

\newcommand{\real}{\mathbb{R}}
\newcommand{\interval}{F}

\newcommand{\Cpw}{\mathcal{C}_{\mathrm{pw}}^1}
\newcommand{\Lip}{\mathcal{L}}
\newcommand{\Cdiff}{\mathcal{C}}

\newcommand{\ima}{\operatorname{Im}}
\newcommand{\Int}{\operatorname{Int}}

\newcommand{\krel}{d^+}
\newcommand{\Krel}{K^+}
\newcommand{\kurv}{\tilde{d}^+}
\newcommand{\Kurv}{\tilde{K}^+}

\newcommand{\ICpw}{\check{I}^+}
\newcommand{\ICpwp}{\check{I}^+}
\newcommand{\ICpwpm}{\check{I}^\pm}
\newcommand{\ICpwm}{\check{I}^-}
\newcommand{\ILip}{{I}^+}
\newcommand{\Ialt}{\mathcal{I}^+}

\newcommand{\Ialtm}{\mathcal{I}^-}
\newcommand{\Ialtpm}{\mathcal{I}^\pm}
\newcommand{\cmark}{\ding{51}}
\newcommand{\xmark}{\ding{55}}

\theoremstyle{plain}
\newtheorem{thm}{Theorem}[section]
\newtheorem{conj[thm]}{Conjecture}
\newtheorem{lem}[thm]{Lemma}
\newtheorem{prop}[thm]{Proposition}

\theoremstyle{definition}
\newtheorem{defn}[thm]{Definition}
\newtheorem{ex}[thm]{Example}

\newtheorem{rem}[thm]{Remark}

\title{Causality theory of spacetimes with continuous Lorentzian metrics revisited}
\author{Leonardo Garc\'ia-Heveling
 \thanks{Department of Mathematics, Radboud University Nijmegen, The Netherlands. \\ \textit{Email:} \tt{l.heveling@math.ru.nl} \\ \textit{Acknowledgements:} I am very grateful to Annegret Burtscher for discussions and detailed comments on the draft. I also wish to thank Klaas Landsman, Ettore Minguzzi and two anonymous referees for further comments.}}
\date{}

\begin{document}

\maketitle

\begin{abstract}
 We consider the usual causal structure $(I^+,J^+)$ on a spacetime, and a number of alternatives based on Minguzzi's $D^+$ and Sorkin and Woolgar's $K^+$, in the case where the spacetime metric is continuous, but not necessarily smooth. We compare the different causal structures based on three key properties, namely the validity of the push-up lemma, the openness of chronological futures, and the existence of limit causal curves. Recall that if the spacetime metric is smooth, $(I^+,J^+)$ satisfies all three properties, but that in the continuous case, the push-up lemma fails. Among the proposed alternative causal structures, there is one that satisfies push-up and open futures, and one that has open futures and limit curves. Furthermore, we show that spacetimes with continuous metrics do not, in general, admit a causal structure satisfying all three properties at once.
 \textit{MSC-classification:} 53C50 (Primary) 83C75 (Secondary)
\end{abstract}

\section{Introduction}

The study of spacetimes with metrics of low regularity is a topic of rising importance in Lorentzian geometry. The main motivation stems from the strong cosmic censorship conjecture \cite{DaLu,Sbi} and the occurrence of weak solutions to Einstein's equations coupled to certain matter models \cite{BuLF,GeTr}. It has hence become an important research question to establish which properties of the usual, smooth spacetimes are more ``robust'' or ``fundamental'', in the sense that they continue to hold in lower regularity, and which, on the other hand, depend sensibly on the smoothness assumption. In trying to answer this question, the need arises to axiomatize the notion of spacetime, and in particular, to treat the causal structure in an order-theoretic way. This, in turn, connects well with ideas in quantum gravity, such as causal set theory \cite{BLMS}. To make matters more concrete, in the present paper we shall study spacetimes $(M,g)$ where $g$ is a continuous Lorentzian metric. However, since our approach is indeed of the order-theoretic type, it can easily be adapted to other settings.

Let us start by recalling the case of a classical spacetime $(M,g)$ where $g$ is smooth. The chronological and causal relations $I^+$ and $J^+$ can then be defined using the notions of timelike and non-spacelike curve respectively. The three following facts are well-known:
\begin{enumerate}
 \item The push-up lemma: if $p \in I^+(q)$ and $q \in J^+(r)$ then $p \in I^+(r)$.
 \item The limit curve theorem: the uniform limit of a converging sequence of causal curves is a causal curve.
 \item Openness of chronological pasts and futures: the sets $I^\pm(p)$ are open, for any $p \in M$.
\end{enumerate}
With these three results at hand, one can develop a large portion of causality theory, such as the causal ladder and the characterization of time functions, without ever again mentioning the metric $g$ or the manifold structure on $M$ explicitly. This is confirmed by the ``Lorentzian length spaces'' approach of Kunzinger and S\"{a}mann \cite{KuSa} and follow-up work \cite{ACS,Leo}.

In order to do causality theory on a spacetime with a $\Cdiff^0$-metric, the first question is how the causal structure should even be defined. The obvious answer is to define $I^+$ and $J^+$ through timelike and non-spacelike curves, just as in the smooth case. However, there are two potential problems:
\begin{enumerate}[label=\Alph*.]
 \item Points where the metric is not $\Cdiff^2$ do not admit normal neighborhoods.
 \item The regularity class where we define timelike curves becomes important.
\end{enumerate}
Chru\'{s}ciel and Grant \cite{ChGr} showed that because of A, the push-up lemma fails, while the limit curve theorems are unaffected (points 1 and 2 above). As a consequence, spacetimes with continuous metric exhibit so called ``causal bubbles'', open regions contained in $J^+$ but not in $I^+$. Regarding B, when the metric is at least $\Cdiff^2$, it was shown by Chru\'{s}ciel \cite{Chr} that one obtains the same chronological relation $I^+$ regardless of whether timelike curves are required to be Lipschitz or piecewise-differentiable. In the case of continuous metrics, however, it was shown by Grant et al. \cite{GKSS} that this choice makes an important difference. In particular, they showed that the chronological futures and pasts are open when using piecewise-differentiable curves, but not when using Lipschitz curves (point 3 above).

A radically different, and in fact earlier, approach is that of Sorkin and Woolgar \cite{SoWo}. They propose to keep the definition of $I^+$ by piecewise-differentiable timelike curves, and then introduce another relation $K^+$ as the smallest transitive, topologically closed relation containing $I^+$. The relation $K^+$ can then be used to replace $J^+$. Even for smooth metrics, the two relations $K^+$ and $J^+$ do not coincide. Nonetheless, it is possible to define the usual causal curves (and hence $J^+$) in terms of $K^+$, without referring to the metric directly. The $K^+$-relation has since found a variety of applications, most notably Minguzzi's works on stable causality \cite{Minstab} and time functions \cite{Minutil}. However, there is no push-up lemma for the $K^+$-relation.

Following a similar philosophy, one can define a relation $\krel$ as the largest relation such that the push-up lemma holds true. On spacetimes with $\Cdiff^2$-metrics, it was shown by Minguzzi \cite{MinD} that such a maximal relation exists, and various characterizations of it are given, using the name $D^+$. We will discuss which of these characterizations are valid for continuous metrics (not all of them, hence the change in notation). The $D^+$-relation was used by Minguzzi in order to characterize the causality condition known as weak distinction: it is shown that a spacetime is weakly distinguishing if and only if $(I^+,D^+)$ is a causal structure in the sense of Kronheimer and Penrose \cite{KrPe}. This result also holds in the $\Cdiff^0$-setting for $d^+$, as we will see in detail. Further, we propose a definition of causal curve based on $\krel$, similar to Sorkin and Woolgar's $\Krel$-causal curves. When the metric is smooth, causal curves defined through $\krel$ coincide with those defined by the metric $g$, but when the metric is merely continuous, they do not. As a consequence, we show that while our new causal relation satisfies the push-up lemma, the limit curve theorems cease to hold. We then argue that, essentially because we chose $\krel$ to be maximal, there in fact do not exist any causal relations that can satisfy both properties at the same time. That is, at least, if we want to keep the usual definition of (piecewise continuously differentiable) timelike curve, the only one that guarantees open futures. We also explore the possibility of alternative chronological relations, but we conclude that one runs into the same problems.

We conclude that spacetimes with continuous metrics are unavoidably pathological. This strengthens the view, already present in the literature, that one should focus on a special class of continuous metrics, the so-called causally plain ones. These are the ones where the usual push-up lemma holds, and include, for example, the class of locally Lipschitz metrics, but not the class of H\"older continuous metrics \cite[Thm. 1.20]{ChGr}. In this sense, our paper can be seen to support a $\Cdiff^{0,1}_\mathrm{loc}$-formulation of strong cosmic censorship, such as in Sbierski's recent work \cite{Sbi2}. Moreover, the methods and results of this paper are also highly relevant to the study of axiomatic causality relations in other settings; for example, in the study of spacetimes with degenerate metrics \cite{DGS}, or with metrics that are not even continuous \cite{GeTr}.

\textbf{Outline.} In Section \ref{seckrel} we provide more background, define the new causal relation $\krel$ and study its properties. In Section \ref{seckurv} we define causal curves in terms of $\krel$, and show that these are just the usual causal curves when the metric is smooth. In Section \ref{other}, we discuss other possible choices of causal structure. In Section \ref{conc}, we summarize and discuss our results.

\section{The $\krel$-relation} \label{seckrel}

\subsection{Basic notions in causality theory} \label{basicdefs}

Let $M$ denote a Hausdorff, paracompact $\Cdiff^1$-manifold, and $g$ a $\Cdiff^0$-Lorentzian metric. Assume that $(M,g)$ admits a $\Cdiff^0$-vector field $X$ such that $g(X,X)<0$, called a time orientation. The pair $(M,g)$ together with a choice of time orientation is called a \emph{$\Cdiff^0$-spacetime}, following the nomenclature of Ling \cite{Ling}. Whenever we say that $g$ is $\Cdiff^2$ (or smooth), or that $(M,g)$ is a $\Cdiff^2$ (or smooth) spacetime, we mean that $M$ admits a $\Cdiff^3$ (resp. smooth) subatlas such that $g$ is $\Cdiff^2$ (resp. smooth) in this subatlas. By \emph{relation} on $M$ we will mean a subset of $M \times M$. The closure of a relation $R$, denoted by $\overline{R}$, is the topological closure in the product topology on $M \times M$. Likewise, we say that $R$ is open if it is an open subset of $M \times M$.

There exist two different definitions for the notion of \emph{timelike curve}, and thus two different notions of \emph{chronological relation}. For $\Cdiff^2$-metrics the two notions are equivalent \cite[Cor.\ 2.4.11]{Chr}, but not for $\Cdiff^0$-metrics \cite{GKSS}. One is based on the class $\Lip$ of locally Lipschitz curves, and the other one on the class $\Cpw$ of piecewise continuously differentiable curves:
\begin{align*}
 \ILip := \{ (p,q) \in M \times M \mid \ &\text{there exists an $\Lip$-curve $\gamma : [a,b] \to M$ such that $\gamma(a)=p$, $\gamma(b) = q$,} \\ & \text{and $g(\dot{\gamma},\dot{\gamma}) < 0$, $g(\dot{\gamma},X) < 0$ almost everywhere} \}, \\
 \ICpw := \{ (p,q) \in M \times M \mid \ &\text{there exists a $\Cpw$-curve $\gamma : [a,b] \to M$ such that $\gamma(a)=p$, $\gamma(b) = q$,} \\ & \text{and $g(\dot{\gamma},\dot{\gamma}) < 0$, $g(\dot{\gamma},X) < 0$ everywhere} \}.
\end{align*}
Recall that an $\Lip$-curve is differentiable almost everywhere by Rademacher's theorem. In the second case, when $\gamma$ is $\Cpw$, the condition $g(\dot{\gamma},\dot{\gamma}) < 0$ is meant to hold for both one-sided derivatives (which may differ from each other at break points). It was established by Grant et al. that $\ICpw$ is open, but $\ILip$ is not necessarily open \cite{GKSS}. The standard \emph{$g$-causal relation} is defined as
\begin{align*}
 J^+ := \{ (p,q) \in M \times M \mid \ &\text{there exists an $\Lip$-curve $\gamma : [a,b] \to M$ such that $\gamma(a)=p$, $\gamma(b) = q$,} \\ & \text{and $g(\dot{\gamma},\dot{\gamma}) \leq 0$, $g(\dot{\gamma},X) < 0$ almost everywhere} \},
\end{align*}
where we say that $\gamma$ is a \emph{$g$-causal curve}. We can analogously define the past relations $I^-$, $\check{I}^-$ and $J^-$ by requiring $g(\dot{\gamma},X) > 0$, but since they are simply given by reversing the factors, there is no need to treat them separately. Therefore we also shall not specify every time that our timelike and causal curves are always future-directed. We do, however, sometimes use the notations $q \in J^+(p)$, and $p \in J^-(q)$, both meaning the same, namely $(p,q) \in J^+$. As mentioned in the introduction, Sorkin and Woolgar suggested an order-theoretic alternative $\Krel$ to $J^+$.

\begin{defn}[{\cite[Def.\ 8]{SoWo}}]
 The relation $\Krel$ is the smallest closed, transitive relation containing $\ICpw$.
\end{defn}

We finish this subsection with a short digression about limit curve theorems. In the literature, there exist multiple statements with this name; a detailed review can be found in \cite{Minlimi} for smooth spacetimes. In \cite[Thm.\ 1.6]{ChGr} it is shown how the limit curve theorems for the smooth case also carry over to continuous spacetimes (see also \cite[Thm.\ 1.5]{Sae}). Rougly speaking, a limit curve theorem is the combination of the following two statements:
\begin{enumerate}
 \item Under certain assumptions, a sequence of causal curves has a convergent (in some appropriate sense) subsequence.
 \item The limit of said subsequence is itself a causal curve.
\end{enumerate}
Here causal usually means $g$-causal, but we will also discuss alternative notions of causal curve. Regarding part 1, there exist many versions tailored to different applications. A common variation is to require the curves to be Lipschitz (as we did in our definition of $J^+$) and add some compactness assumptions in order to apply the Arzel\`{a}--Ascoli theorem. Part 2 is where the causal structure becomes important; we discuss it in the context of our new $\krel$-relation in Remark \ref{limcurvthm}.

\subsection{Definition of $\krel$}

We first introduce some nomenclature, the underlying concepts being fairly standard. By $(M,g)$ we continue to denote a $\Cdiff^0$-spacetime, although some of the ideas make sense even if $M$ is just a set. The following definition gives a compatibility condition between the chronological and causal relations, in this case denoted abstractly by $R$ and $S$ respectively.
\begin{defn} \label{pushup}
 Let $R,S \subseteq M \times M$ be two relations. We say that $S$ satisfies \emph{push-up relative to $R$} (or that \emph{$R$ is an $S$-ideal}, cf.\ \cite{Min,MinD}) if the following two properties hold:
 \begin{enumerate}[label=(\roman*)]
  \item $(x,y) \in S, (y,z) \in R \implies (x,z) \in R$,
  \item $(w,x) \in R, (x,y) \in S \implies (w,y) \in R$.
 \end{enumerate}
\end{defn}
Let $R,S,S' \subseteq M \times M$ be relations. Then it is easy to see that:
 \begin{enumerate}[label=(\alph*)]
  \item If $R$ is transitive, then $R$ satisfies push-up relative to itself.
  \item If $S$ satisfies push-up relative to $R$, and $S' \subseteq S$, then also $S'$ satisfies push-up relative to $R$.
  \item If $S$ and $S'$ each satisfy push-up relative to $R$, then so does $S \cup S'$.
 \end{enumerate}
If $(M,g)$ is $\Cdiff^2$, then $J^+$ satisfies push-up relative to $\ILip$ (and equivalently, $\ICpw$). This fact is known as the push-up lemma \cite[Lem.\ 2.4.14]{Chr}. Those $\Cdiff^0$-spacetimes where $J^+$ satisfies push-up relative $\ICpw$ are called \emph{causally plain}. The term was coined by Chru\'{s}ciel and Grant \cite[Def.\ 1.16]{ChGr}, but beware that they used $\ILip$ in place of $\ICpw$. In any case, not all $\Cdiff^0$-spacetimes are causally plain \cite[Ex.\ 1.11]{ChGr}. The failure of the push-up lemma on arbitrary $\Cdiff^0$-spacetimes motivates our next, central definition.

\begin{defn} \label{kdef}
 The \emph{$\krel$-relation} is the largest relation that satisfies push-up relative to $\ICpw$.
\end{defn}

\begin{prop} \label{welldef}
 There exists a unique relation $\krel$ satisfying Definition \ref{kdef}. Moreover, $\krel$ is transitive and reflexive.
\end{prop}

\begin{proof}
 We construct $\krel$ by setting $(x,y) \in \krel$ if:
 \begin{enumerate}[label=(\roman*)]
  \item $(y,z) \in \ICpw \implies (x,z) \in \ICpw$ (in other words, $\ICpwp(y) \subseteq \ICpwp(x)$),
  \item $(w,x) \in \ICpw \implies (w,y) \in \ICpw$ (in other words, $\ICpwm(x) \subseteq \ICpwm(y)$).
 \end{enumerate}
 Clearly, this relation is maximal such that Definition \ref{pushup} is satisfied. Moreover, for any $p \in M$, $(p,p)$ satisfies (i) and (ii), meaning that $\krel$ is reflexive. To show transitivity, assume $(x,y) \in \krel$ and $(y,z) \in \krel$. Then $\ICpwp(z) \subseteq \ICpwp(y) \subseteq \ICpwp(x)$ and $\ICpwm(x) \subseteq \ICpwm(y) \subseteq \ICpwm(z)$ , hence $(x,z) \in \krel$.
\end{proof}
 
For spacetimes with $\Cdiff^2$-metrics, the existence of such a maximal relation (i.e., of $\krel$ according to our notation) was already noted by Minguzzi \cite{MinD}, who defined
\begin{equation} \label{defD}
  D^+ := \left\{ (p,q) \in M \times M \mid q \in \overline{\ICpwp}(p), \ p \in \overline{\ICpwm}(q) \right\},
 \end{equation}
maximality with respect to push-up being a consequence (rather than the defining property) in this case \cite[Lem.\ 2.8]{MinD}. In the $\Cdiff^2$- (but \emph{not} in the $\Cdiff^0$-) case, \eqref{defD} is equal to
\[
 \left\{ (p,q) \in M \times M \mid q \in \overline{J^+}(p), \ p \in \overline{J^-}(q) \right\},
\]
which is subsequently taken to be the definition in the review article \cite[Sec.\ 4.1]{Min}. Note also that in the $\Cdiff^2$-case, the distinction between $\ILip$ and $\ICpw$ is irrelevant. To avoid possible confusion, we adopt the lower-case notation $\krel$, although we now prove that $\krel = D^+$ also on $\Cdiff^0$-spacetimes.

\begin{lem} \label{kinI}
 For $D^+$ as given by \eqref{defD}, it holds that $\krel = D^+$. In particular, $\krel \subseteq \overline{\ICpw}$.
\end{lem}

\begin{proof}
 Suppose $(p,q) \in \krel$. Let $\gamma : [0,1) \to M$ be any timelike curve with $\gamma(0) = q$. Then for all $t \in (0,1)$, $\left(q,\gamma(t)\right) \in \ICpw$. Because $(p,q) \in \krel$, the push-up property implies $\left(p,\gamma(t)\right) \in \ICpw$. Since $\gamma$ is continuous, $(p,\gamma(t)) \to (p,q)$ as $t \to 0$, hence $q \in \overline{\ICpwp}(p)$. Similarly, $p \in \overline{\ICpwm(q)}$, so we conclude that $(p,q) \in D^+$ and hence $\krel \subseteq D^+$.
 
 On the other hand, suppose that $(p,q) \in D^+$, and let $r \in \ICpwp(q)$. By assumption, $q \in \overline{\ICpwp}(p)$, so we can approximate $q$ by a sequence $(q_i)_i$ in $\ICpwp(p)$. By openness of $\ICpw$, $r \in \ICpwp(q_i)$ for some large enough $i$. Then, by transitivity, we conclude that $r \in \ICpwp(p)$. Since $r$ was arbitrary, we conclude that $\ICpwp(q) \subseteq \ICpwp(p)$. Similarly, $\ICpwm(p) \subseteq \ICpwm(q)$, so $D^+$ satisfies push-up with respect to $\ICpw$, and hence $D^+ \subseteq d^+$.
 
 Finally, note that $D^+ \subseteq \overline{\ICpw}$ as a direct consequence of \eqref{defD}.
\end{proof}

We end this subsection by showing how the $\krel$-relation fits into the current literature. A $\Cdiff^0$-spacetime $(M,g)$ is said to be \emph{weakly distinguishing} whenever, for all $p,q \in M$, $\ICpwp(p) = \ICpwp(q)$ and $\ICpwm(p) = \ICpwm(q)$ together imply $p=q$ \cite[Def.\ 4.47]{Min}. Given two relations $R,S$ on $(M,g)$ (or any set, for that matter), we say that the pair $(R,S)$ is a \emph{causal structure} in the sense of Kronheimer and Penrose \cite[Def.\ 1.2]{KrPe} if:
\begin{enumerate}[label=(\roman*)]
 \item $S$ is transitive, reflexive and antisymmetric,
 \item $R$ is contained in $S$ and irreflexive,
 \item $S$ satisfies push-up relative to $R$.
\end{enumerate}
The following proposition was already shown by Minguzzi for $\Cdiff^2$-spacetimes (phrased in terms of $D^+$ and with a slightly different proof, see \cite[Thm.\ 4.49]{Min} or \cite[Lem.\ 2.7]{MinD}).

\begin{prop} \label{KP}
 Let $(M,g)$ be a $\Cdiff^0$-spacetime. Then $(\ICpw,\krel)$ is a causal structure in the sense of Kronheimer and Penrose if and only if $(M,g)$ is weakly distinguishing.
\end{prop}

\begin{proof}
 Point (iii) is satisfied by the very definition of $\krel$. To see point (ii), recall that if $M$ is weakly distinguishing, then $M$ is chronological, i.e., it does not contain closed timelike curves (and hence $\ICpw$ is irreflexive). Indeed, suppose $\gamma : [a,b] \to M$ is a closed timelike curve. Then, by transitivity, $\ICpwpm(\gamma(t)) = \ICpwpm(\gamma(s))$ for all $s,t \in [a,b]$. On the other hand, $\gamma$ must be non-constant, in contradiction to weak distinction. That $\ICpw$ is contained in $\krel$ is clear because $\ICpw$, being transitive, must satisfy push-up with respect to itself . Regarding point (i), since $\krel$ is always transitive and reflexive by Lemma \ref{welldef}, it only remains to show that $\krel$ is antisymmetric.
 
 By the proof of Lemma \ref{welldef}, $(p,q) \in \krel$ if and only if $\ICpwp(q) \subseteq \ICpwp(p)$ and $\ICpwm(p) \subseteq \ICpwm(q)$. Hence, $(p,q)\in\krel$ and $(q,p) \in \krel$ if and only if $\ICpwp(p) = \ICpwp(q)$ and $\ICpwm(p) = \ICpwm(q)$. But $p=q$ for all such pairs $(p,q)$ if and only if $(M,g)$ is weakly distinguishing.
\end{proof}

Another way of phrasing the last result in the usual language of causality theory is to say that ``$(M,g)$ is $\krel$-causal if and only if it is weakly distinguishing''.

\subsection{The local structure of $\krel$}

Given a neighborhood $U \subseteq M$, we can define the localized relations $\ICpw_U$, $J^+_U$ and $\krel_U$  by applying the usual definitions to the spacetime $(U,g\vert_U)$. It is easy to see that $\ICpw_U \subseteq \ICpw $ and $J^+_U \subseteq J^+$. 

\begin{lem} \label{kUink}
 Let $U \subseteq M$ be an open neighborhood. Then $\krel_U \subseteq \krel$.
\end{lem}

\begin{proof}
 By definition, $\krel_U$ satisfies push-up  relative to $\ICpw_U$. We need to show that $\krel_U$ also satisfies push-up relative to $\ICpw$, and then the claim follows from maximality of $\krel$. Suppose that $(x,y) \in \krel_U$ and $(y,z) \in \ICpw$. Then there exists a timelike curve $\gamma : [0,1] \to M$ from $y$ to $z$. Since, by assumption, $y \in U$, we must have that for $\epsilon > 0$ small enough, $\gamma \vert_{[0,\epsilon]} \in U$. Thus we have $(y, \gamma(\epsilon)) \in \ICpw_U$, which implies $(x,\gamma(\epsilon)) \in \ICpw_U \subseteq \ICpw$ by definition of $\krel_U$. Since also $(\gamma(\epsilon), z) \in \ICpw$, transitivity of $\ICpw$ implies $(x,z) \in \ICpw$. Part (ii) of Definition \ref{pushup} can be shown analogously.
\end{proof}

We want to investigate whether, for $U$ small enough, $\krel_U$ is closed. The motivation lies in the limit curve theorems (see Remark \ref{limcurvthm} for the details). By Lemma \ref{kinI}, $\krel_U \subseteq \overline{\ICpw_U}$. Since also $\ICpw_U \subseteq \krel_U$, we conclude that $\krel_U$ is closed if and only if $\krel_U = \overline{\ICpw_U}$. By Definition \ref{kdef}, $\krel_U = \overline{\ICpw_U}$ if and only if $\overline{\ICpw_U}$ satisfies push-up. Unfortunately, the next example \cite[Ex.\ 3.1]{GKSS}, shows that the latter is not necessarily the case.

\begin{ex} \label{exnotclos}
 Let $M = \real^2$ with metric given by
 \[
  g_\alpha := -\sin 2\theta(x) dt^2 -2\cos 2\theta(x) dx dt + \sin 2\theta(x) dx^2
 \]
 where
 \[
  \theta(x) := \begin{cases}
                  0, &x < -1 \\
                  \arccos \vert x \vert^\alpha, &-1 \leq x \leq 0 \\
                  \frac{\pi}{2}, & x > 0,
                 \end{cases}
 \]
 and $0 < \alpha < 1$ arbitrary. The metric $g_\alpha$ is $\alpha$-Hölder continuous, and in fact smooth outside of $\{x=-1\} \cup \{x=0\}$. This example was introduced by Grant et al. \cite[Ex.\ 3.1]{GKSS}, who showed that $\ICpw \subsetneq \ILip$.
 
 Let $p = (0,0)$ and $U \subseteq M$ any open neighborhood of $p$. Then the following hold:
 \begin{enumerate}[label=(\roman*)]
  \item $\ILip_U$ is not open,
  \item $\krel_U$ is not closed.
 \end{enumerate}
 Point (i) is shown in \cite[Ex.\ 3.1]{GKSS} (they in fact show that $\ILip$ is not open, but their argument is also valid on neighborhoods).
 
 In order to show point (ii), first note that the past $\ICpwm(p)$ (blue region in Figure \ref{fig1}) is contained in $\{x>0\}$. This is so because a timelike $\Cpw$-curve must have timelike tangent vector everywhere, which implies having a positive $x$-component when in $\{ x \geq 0 \}$. When using $\Lip$-curves, the past set $I^-(p)$ is in fact bigger \cite[Ex.\ 3.1]{GKSS}, but we will not discuss this further.
 
 Consider the curve $\gamma : (- \epsilon,0) \to M, s \mapsto (t(s),x(s))$ given by
 \begin{align*}
  &t(s) := \frac{1}{1-\alpha} A^{1-\alpha} s , &x(s) := - A \vert s \vert^{\frac{1}{1-\alpha}} ,
 \end{align*}
 where $A>0$ is arbitrary and $\epsilon > 0$ is small. It is shown in \cite[Ex.\ 3.1]{GKSS} that $\gamma$ is a timelike curve in the $\Cpw$-sense (but its extension to the endpoint $s=0$ is not). Since $\gamma(s) \to p$ as $s \to 0$, we conclude that $(\gamma(-\epsilon'),p) \in \overline{\ICpw}$ for all $\epsilon' < \epsilon$.
 
 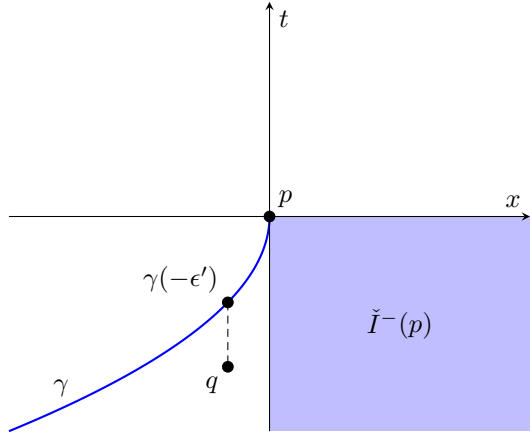
\begin{figure}
  \begin{center}
   \begin{tikzpicture}
    \begin{axis}[xmin=-1,xmax=1,
                 ymin=-2,ymax=2,
                 axis on top=true,
                 axis x line=middle,
                 axis y line=middle,
                 xlabel={$x$},
                 ylabel={$t$},
                 xtick=\empty,
                 ytick=\empty,]
     \fill[fill=blue!25] (axis cs: 0,0) rectangle (axis cs: 1,-2);
     \addplot[color=blue,style=thick,domain=-0.5:0,samples=50]({-4*x^2},{4*x});
     \filldraw[black] (axis cs: 0,0) circle (2pt) node[anchor=south west] {$p$};
     \filldraw[black] (axis cs: -0.16,-0.8) circle (2pt) node[anchor=south east] {$\gamma(-\epsilon')$};
     \filldraw[black] (axis cs: -0.16,-1.4) circle (2pt) node[anchor=north east] {$q$};
     \node at (axis cs: 0.5,-1) {$\ICpwm(p)$};
     \node at (axis cs: -0.8,-1.6) {$\gamma$};
     \draw[densely dashed] (axis cs: -0.16,-0.8) -- (axis cs: -0.16,-1.4);
    \end{axis}
   \end{tikzpicture}
  \end{center}
  \caption{The spacetime of Example \ref{exnotclos}, with the curve $\gamma$ that lies outside of $\ICpwm(p)$, while nonetheless $(\gamma(-\epsilon'),p) \in \overline{\ICpwp}$.} \label{fig1}
 \end{figure}

 Next we show that $(\gamma(-\epsilon'),p) \not\in \krel$. To see this, consider any point of the form $q = (x(-\epsilon'),t_q)$ with $t_q < t(-\epsilon')$ (see Figure \ref{fig1}). Then $(q,\gamma(-\epsilon')) \in \ICpw$, the connecting vertical segment being an example of timelike $\Cpw$-curve between $q$ and $\gamma(-\epsilon')$. If we assume $(\gamma(-\epsilon'),p) \in \krel$, then by push-up it follows that $(q,p) \in \ICpw$. However, $(q,p) \not\in \ICpw$ because $x(-\epsilon') < 0$ and $\ICpwm(p)$ is contained in $\{ x>0 \}$. Hence $(\gamma(-\epsilon'),p) \not\in \krel$.
 
 Since we can pick $\epsilon'>0$ arbitrarily small, and $t_q$ smaller but arbitrarily close to $t(-\epsilon')$, the previous discussion applies to any neighborhood $U$ of $p$. Thus $\krel_U \subsetneq \overline{\ICpw}_U$, no matter how we choose $U$.
\end{ex}

Grant et al. showed that in the previous example also $\ICpw \subsetneq \ILip$, and that $\ILip$ is not open (recall that $\ICpw$ is always open). If we were to define $\krel$ by requiring push-up with respect to $\ILip$ instead of $\ICpw$, it may be that $\krel_U$ is closed (for small enough $U$). We do not explore this possibility here, and simply note that this would be at the cost of chronological futures not being open. Hence the conclusion is, either way, that one cannot have push-up, open futures and (local) closedness at the same time. That is, at least, if one wants the chronological relation to be given by the usual $\ILip$ or $\ICpw$. We explore alternatives to $\ILip$ and $\ICpw$ in Section \ref{other}, but the conclusion there is also that one of the three properties has to be sacrificed.

\section{Causal curves in terms of the $\krel$-relation} \label{seckurv}

In this subsection, we define a variation on the $\krel$-relation, which we call $\kurv$, based on based on the notion of \emph{$\krel$-causal curves}. The motivation for this is two-fold. Firstly, we will show that $\kurv$, unlike $\krel$, has the desirable property that $\kurv = J^+$ on smooth spacetimes. Secondly, we are interested in studying the potential validty of limit curve theorems for $\krel$-causal curves. Throughout this section, $\interval$ denotes an interval, meaning any connected subset of $\real$.

\begin{defn} \label{defkloc}
 A continuous curve $\gamma : \interval \to M$ is called \emph{$\krel$-causal} if for every $t \in \interval$ and every open neighborhood $U \subseteq M$ of $\gamma(t)$, there exists an open neighborhood $V \subseteq \interval$ of $t$ such that
 \[
  s_1 < s_2 \implies \left(\gamma(s_1),\gamma(s_2)\right) \in \krel_U \ \text{ for all } \ s_1,s_2 \in V.
 \]
\end{defn}

\begin{rem} \label{gJcausal}
 Similarly to Definition \ref{defkloc}, one can define $J^+$-causal curves, as was done already by Hawking and Ellis in 1973 \cite[Chap.\ 6.2]{HaEl}, and $\Krel$-causal curves \cite[Def.\ 17]{SoWo}. Any $g$-causal curve is automatically also $J^+$-causal. However, the converse is not true, since a $J^+$-causal curve may not even have a well defined tangent vector. Nonetheless, if two points $p,q \in M$ can be joined by a $J^+$-causal curve $\gamma$, then $(p,q) \in J^+$. In particular, there exists a $g$-causal curve $\sigma$, not necessarily equal to $\gamma$, which joins $p$ and $q$.
\end{rem}

Similarly to the previous remark, by transitivity and Lemma \ref{kUink} it follows that if two points $p,q \in M$ can be joined by a $\krel$-causal curve, then $(p,q) \in \krel$. The next example motivates why Definition \ref{defkloc} has to be formulated in a local way, i.e.\ why we do not simply require $s_1 < s_2 \implies \left(\gamma(s_1),\gamma(s_2)\right) \in \krel$ for all $s_1,s_2 \in F$.

\begin{ex} \label{notsame}
 Let $M = S^1 \times \real$ with metric $ds^2 = -dt^2 + dx^2$. This spacetime is totally vicious, so $\ICpw = M \times M$ and hence also $\krel = M \times M$. Therefore, any $\Cdiff^0$-curve $\gamma: F \to M$ satisfies $\left(\gamma(s_1),\gamma(s_2)\right) \in \krel$ for all $s_1,s_2 \in F$. However, $M$ locally looks like Minkowski spacetime, where $\krel = J^+$, hence not all curves on $M$ are $\krel$-causal in the sense of Definition \ref{defkloc}.
\end{ex}

Having defined $\krel$-causal curves, we briefly return to Example \ref{exnotclos} in order to better understand the relationship between closedness of $\krel$ and limit curve theorems.

\begin{rem}[On limit curve theorems] \label{limcurvthm}
 Suppose that $(\gamma_n)_n$ is a sequence of $\krel$-causal curves converging pointwise to a $\Cdiff^0$-curve $\gamma_\infty: \interval \to M$. Suppose that for every $t \in \interval$, there exists a neighborhood $U \subseteq M$ of $\gamma_\infty(t)$ such that $\krel_U$ is closed. Then the curve $\gamma_\infty$ is $\krel$-causal, because any pair of points on $\gamma_\infty$ can be written as a limit of pairs of points on $\gamma_n$.
 
In Example \ref{exnotclos}, we showed that the point $p=(0,0)$ does not admit any neighborhood $U$ such that $\krel_U$ is closed. Let $\gamma$ be as in Example \ref{exnotclos}, and consider the sequence of curves given by $\gamma_n = \gamma \vert_{(-\epsilon, 1/n]}$. The sequence $(\gamma_n)_n$ converges pointwise (even uniformly, after appropriate reparametrization) to a curve $\gamma_\infty : (-\epsilon,0] \to M$ which is just $\gamma$ with $p$ added as its endpoint. However, we showed in Example \ref{exnotclos} that $(\gamma(t),p) \not\in \krel$ for all $-\epsilon < t < 0$, hence $\gamma_\infty$ is not $\krel$-causal.
\end{rem}

Moving on, we use $\krel$-causal curves to define a new causal relation on $M$.

\begin{defn}
 We define the $\kurv$-relation as follows: $(p,q) \in \kurv$ if there exists a $\krel$-causal curve from $p$ to $q$.
\end{defn}

By Lemma \ref{kUink} and transitivity of $\krel$, $\kurv \subseteq \krel$. In particular, $\kurv$ satisfies push-up relative to $\ICpw$. It is also clear that the concatenation of two $\krel$-causal curves is again $\krel$-causal, hence $\kurv$ is transitive. Example \ref{notsame} shows that $\kurv$ can be strictly smaller than $\krel$. We finish this section with one of the main results of the paper, namely that $\kurv = J^+$ on smooth (and more generally, causally plain) spacetimes. Recall that a $\Cdiff^0$-spacetime is called causally plain if $J^+$ satisfies push-up relative to $\ICpw$.

\begin{lem} \label{kurvinJ}
 Let $(M,g)$ be a $\Cdiff^0$-spacetime, and $\gamma : \interval \to M$ a $\krel$-causal curve. Then $\gamma$ is also $J^+$-causal.
\end{lem}

\begin{proof}
 Let $t \in \interval$ be arbitrary. By \cite[Proposition 1.10]{ChGr}, there exists a neighborhood $U$ of $p := \gamma(t)$, called a cylindrical neighborhood, such that $\overline{\ICpwpm_U(p)} \subseteq J_U^\pm(p)$ (here we mean the closure of the set $\ICpwpm_U(p) \subseteq U$). Let $V \in \interval$ be a neighborhood of $t$ as in Definition \ref{defkloc}, and $s \in V$. Suppose $s \geq t$, the other case being analogous. Because $\gamma$ is $\krel$-causal, we have $\left(p,\gamma(s)\right) \in \krel_U$. Let $\sigma : [0,\epsilon) \to U$ be any timelike $\Cpw$-curve with $\sigma(0) = \gamma(s)$. By push-up, $\ima(\sigma) \subseteq \ICpwp\left(p\right)$. By continuity, $\sigma(u) \to \gamma(s)$ as $u \to 0$. Hence, by the previous and our choice of $U$, $\gamma(s) \in \overline{\ICpwp_U(p)} \subseteq J_U^+(p)$. In other words, $(\gamma(t),\gamma(s)) \in J_U^+$. Since $t, s$ are arbitrary (as long as they are close enough), we conclude that $\gamma$ is $J^+$-causal.
\end{proof}

\begin{thm}
 Let $(M,g)$ be a causally plain $\Cdiff^0$-spacetime. Then $\kurv = J^+$.
\end{thm}

\begin{proof}
 If $(M,g)$ is causally plain, then $J^+$ satisfies push-up, hence $J^+ \subseteq \krel$. In particular, on a subset $U \subset M$, we have $J_U^+ \subseteq \krel_U$. Assume $(p,q) \in J^+$. Then there exists a $g$-causal curve $\gamma:[a,b] \to M$ from $p$ to $q$. By continuity, for every $t \in [a,b]$ and every neighborhood $U$ of $\gamma(t)$, there exists a neighborhood $V \subseteq [a,b]$ of $t$ small enough such that $\gamma \vert_V$ is contained in $U$. If $s_1,s_2 \in V$ and $s_1 < s_2$, then $\left(\gamma(s_1),\gamma(s_2)\right) \in J^+_U \subseteq \krel_U$. Thus $\gamma$ is a $\krel$-causal curve, and since $\gamma$ is arbitrary, we conculde that $J^+ \subseteq \kurv$. The other inclusion follows from Lemma \ref{kurvinJ}, by noting that if two points $p,q$ can be joined by a $J^+$-causal curve, then $(p,q) \in J^+$.
\end{proof}

\section{Other causal structures} \label{other}

It is possible to repeat the procedure of Section \ref{seckurv} for Sorkin and Woolgar's $\Krel$, and define a relation $\Kurv$ based on $\Krel$-causal curves (the latter class of curves is also studied in \cite[Section 3]{SoWo}). On a smooth spacetime, every point admits an arbitrarily small neighborhood $U$ (a convex normal neighborhood) such that $J^+_U = \overline{\ICpw_U}$. Thus we conclude that on smooth spacetimes, $\Kurv = J^+$. Unfortunately, Example \ref{exnotclos} and Remark \ref{limcurvthm} tell us that $\Kurv$ cannot satisfy push-up with respect to $\ICpw$ on all $\Cdiff^0$-spacetimes. That is because, if it did, then $\Kurv \subseteq \krel$. Recall that $\Krel$ is closed and contains $\ICpw$. It follows that if a curve $\gamma$ is the limit of a sequence of $\Cpw$-timelike curves, then $\gamma$ must be $\krel$-causal. However, in Remark \ref{limcurvthm}, we saw an example of such a $\gamma$ where the endpoints are not $\krel$-related to each other, a contradiction.

A different approach is to consider $J^+$ to be the more fundamental relation, and then find an appropriate notion of chronological order, say $\Ialt$. Ideally, we would like all of the following three properties to hold.
\begin{enumerate}[label=(\roman*)]
 \item $\Ialt$ is open and contained in $J^+$.
 \item $J^+$ satisfies push-up relative to $\Ialt$.
 \item For every point $p \in M$ and every neighborhood $U$ of $p$, $\Ialtpm(p) \bigcap U \neq \emptyset$.
\end{enumerate}

\begin{prop}
 Given $J^+$, if there exists a relation $\Ialt$ satisfying (i), (ii) and (iii) above, then $\Ialt = \Int J^+$.
\end{prop}

\begin{proof}
 By property (i), it is clear that $\Ialt \subseteq \Int J^+$. Suppose there exists $(p,q) \in \Int J^+ \setminus \Ialt$. By propery (iii), with $U$ a neighborhood of $q$, we can find $r \in U \cap \Ialtm(q)$. We may choose $U$ small enough so that $\{p\} \times U \subseteq \Int J^+$, meaning in particular that $(p,r) \in J^+$. Then, by property (ii), we have that $(p,q) \in \Ialt$, obtaining a contradiction. Therefore $\Ialt = \Int J^+$.
\end{proof}

Next, we present an example where $J^+$ does not satisfy push-up relative to $\Int J^+$ (point (ii)). In view of the previous proposition, we conclude that there cannot exist a relation $\Ialt$ satisfying (i)-(iii) above.

\begin{ex} \label{exsec4} 
This example is adapted from \cite[Ex.\ 1.11]{ChGr} and \cite[Sec.\ 4.1]{Ling}. Let $M = (-2,2) \times \real$ with metric given by
 \[
  ds^2 = -dt^2 - 2(1- \vert t \vert^{1/2})dtdx + \vert t \vert^{1/2} (2 - \vert t \vert^{1/2})dx^2.
 \]
 This metric is smooth everywhere except on the $x$-axis. A null curve starting at the point $p = (-1,0)$ can be parametrized as $x \mapsto \gamma(x) = (t(x),x)$, and then
 \[
  \dot{t} = \begin{cases}
             \vert t \vert^{1/2} &\text{or} \\
             \vert t \vert^{1/2} - 2. &
            \end{cases}
 \]
 By solving this equation, we obtain the boundary of $J^+(p)$ (light blue region in Figure \ref{fig2}). Consider the first case of the equation, which is when the null curve $\gamma$ moves upwards and to the right. We are interested in finding the value $x_1$ such that $t(x) \to 0$ as $x \to x_1$. It can easily be computed by separation of variables:
 \[
  x_1 = \int^{x_1}_{0} dx = \int^0_{-1} \frac{dt}{\vert t \vert^{1/2}} = 2.
 \]
 Let $q = (0,0)$, then $(p,q) \in \Int J^+$. Consider a third point $r = (0,3)$, as in Figure \ref{fig2}. Now $(q,r) \in J^+$, because the curve $x \mapsto (0,x)$ is null. However, $(p,r) \not\in \Int J^+$, since there are points of the form $(-\epsilon, 3)$ arbitrarily close to $r$ which cannot lie in $J^+(p)$ because they lie below the $x$-axis and their $x$-coordinate is larger than $2$. Hence $J^+$ does not satisfy push-up relative to $\Int J^+$ in this example.
  \begin{figure}
  \begin{center}
   \begin{tikzpicture}
    \begin{axis}[xmin=-4,xmax=4,
                 ymin=-1.5,ymax=1.5,
                 axis on top=true,
                 axis x line=middle,
                 axis y line=middle,
                 xlabel={$x$},
                 ylabel={$t$},
                 xtick=\empty,
                 ytick=\empty,]
     \addplot[color=blue!15,name path =A,style=thick,domain=0:1,samples=50]({2-2*sqrt(x)},{-x});
     \addplot[color=blue!15,name path =B,style=thick,domain=0:1,samples=50]({-2*x-4*ln(2-x)+2},{-(x^2)});
     \addplot[color=blue!15,name path =C,style=thick,domain=0:1.5,samples=50]({2*x+4*ln(2-x)+2-8*ln(2)},{(x^2)});
     \addplot[color=blue!30,name path =D,style=thick,domain=0:1.5,samples=50]({2*x+4*ln(2-x)-4*ln(2)},{(x^2)});
     \path[name path =E] (axis cs: 0,0) -- (axis cs: 4,0);
     \path [name path=F] (axis cs: 0,1.5) -- (axis cs: 4,1.5);
     \path [name path=G] (axis cs: -2,1.5) -- (axis cs: 0,1.5);
     \filldraw[black] (axis cs: 0,0) circle (2pt) node[anchor=south west] {$q$};
     \filldraw[black] (axis cs: 0,-1) circle (2pt) node[anchor=north west] {$p$};
     \filldraw[black] (axis cs: 3,0) circle (2pt) node[anchor=south] {$r$};
     \addplot [blue!15] fill between [of=A and E];
     \addplot [blue!15] fill between [of=B and D];
     \addplot [blue!15] fill between [of=C and D];
     \addplot [blue!30] fill between [of=D and G];
     \addplot [blue!30] fill between [of=E and F];
     \node at (axis cs: 2,0.75) {$J^+(q)$};
     \node[style={rectangle, draw=black}] (jp) at (axis cs: 2,-0.75) {$J^+(p)$};
     \draw[->] (jp.north west) -- (axis cs: 0.6,-0.2);
    \end{axis}
   \end{tikzpicture}

  \end{center}
 \caption{The points $p,q,r$ in Example \ref{exsec4}, which satisfy $(p,q) \in \Int J^+$ and $(q,r) \in J^+$ but $(p,r) \not\in \Int J^+$.} \label{fig2}
 \end{figure}
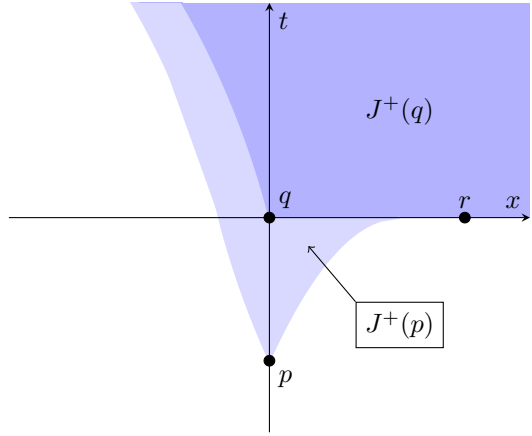
\end{ex}

Finally, we point out yet another option, which is to define chronological futures via the Lorentzian distance. We recall the basic properties of the Lorentzian distance, and refer to \cite[Chap.\ 4]{BEE} for further details. The length of a causal curve $\gamma:[a,b]\to M$ is defined as
\[
 L_g(\gamma) := \int_a^b \sqrt{-g(\dot\gamma,\dot\gamma)} ds,
\]
where it is enough for the integrand to be defined almost everywhere. The Lorentzian distance between two points $p,q \in M$ is then defined by
\begin{equation} \label{defd}
 d(p,q) := \sup \left\{L_g(\gamma) \mid \gamma \text{ a $g$-causal curve with endpoints $p,q$} \right\}
\end{equation}
if $(p,q) \in J^+$, and $d(p,q) := 0$ otherwise. Recall that when the spacetime metric $g$ is smooth, the Lorentzian distance $d$ satisfies $d(p,q) > 0 \iff (p,q) \in \ICpw$ \cite[Eqn.\ 4.2]{BEE}. On $\Cdiff^0$-spacetimes, only the ``$\impliedby$'' implication continues to hold. To see how ``$\implies$'' can fail, take the points $p,q,r$ in Example \ref{exsec4} (depicted in Figure \ref{fig2}). We can connect $p$ and $q$ by a vertical segment, which is timelike (hence causal) and has length equal to $1$. We can connect $q$ and $r$ by a horizontal segment, which is also causal, and has length $0$. Concatenating the two segments, we get a causal curve of length $1$ from $p$ to $r$, despite the fact that $(p,r) \not\in \ICpw$. Further, we see that $r \in \partial J^+(p)$, and since $\{d>0\} \subseteq J^+$ by definition, we conclude that $\{d>0\}$ is not open in this example. Nonetheless, another property of the Lorentzian distance, the inverse triangle inequality
\[
 d(p,r) \geq d(p,q) + d(q,r) \ \text{ if } (p,r),(r,q) \in J^+,
\]
does hold for all $\Cdiff^0$-metrics, since it is a direct consequence of \eqref{defd} and the fact that the concatenation of two causal curves is causal. We deduce from the inverse triangle inequality that $J^+$ satisfies push-up relative to $\{ d>0 \}$. Hence the combination of $J^+$ and $\{d>0\}$ gives us push-up and limit curve theorems, but at the price of non-open futures.

\section{Conclusions} \label{conc}

Table \ref{table} summarizes the properties of different choices of causal and chronological relation on $\Cdiff^0$-spacetimes. While each of the choices (rows) is distinct from the others for $\Cdiff^0$-metrics, they all coincide for smooth metrics. In particular, in the smooth case, the standard causal structure ticks all three boxes. For $\Cdiff^0$-metrics, on the other hand, no combination of chronological and causal order has all of the three properties that we considered. The newly introduced $(\ICpw,\kurv)$ is the only causal structure that has both push-up and an open chronological relation. Moreover, it defines a causal structure in the sense of Kronheimer and Penrose (see Proposition \ref{KP}).

\begin{table}\caption{Comparison of different causal structures on $\Cdiff^0$-spacetimes by three important properties. The (non-)openness of $\ICpw$ and $\ILip$ was established by Grant et al. \cite{GKSS}, and the push-up and limit curve properties of $J^+$ by Chru\'{s}ciel and Grant \cite{ChGr}. The rest of entries in the table were, to the best of our knowledge, never explicitly stated before.} \label{table}
\begin{center}
\begin{tabular}{ | c c | c c c | } 
 \hline
 Chronological & Causal & Push-up & Open & Limit \\
 order & order & & futures & curves \\
 \hline
 & & & & \\[-2ex]
 $\ICpw$ & $J^+$    & \xmark & \cmark & \cmark \\[1ex]
 $\ICpw$ & $\Kurv$  & \xmark & \cmark & \cmark \\[1ex]
 $\ICpw$ & $\kurv$  & \cmark & \cmark & \xmark \\[1ex]
 $\ILip$ & $J^+$    & \xmark & \xmark & \cmark \\[1ex]
 $\ILip$ & $\Kurv$ & \textbf{?} & \xmark & \cmark \\[1ex]
 $\ILip$ & $\kurv$  & \cmark & \xmark & \textbf{?} \\[1ex]
 $\{d>0\}$   & $J^+$    & \cmark & \xmark & \cmark \\[1ex]
 $\Int J^+$ & $J^+$ & \xmark & \cmark & \cmark \\[1ex]
 \hline
\end{tabular}
\end{center}

\end{table}

It is fair to say that we have exhausted all reasonable possibilities. For if we want the chronological relation to be given by timelike $\Cpw$-curves, then Example \ref{exnotclos} and Remark \ref{limcurvthm} tell us that no causal relation can satisfy push-up and at the same time admit a limit curve theorem. We do not know if this changes when using timelike $\Lip$-curves, but even if so, we would loose the openness of chronological futures instead \cite{GKSS}. If, on the other hand, we choose for the causal relation to be given by $g$-causal curves, and try to define a compatible chronological relation, we run into the same problems by the discussion in Section \ref{other}.

\printbibliography

\end{document}